\documentclass[a4paper]{amsart}

\usepackage[margin=2.5cm]{geometry}

\usepackage{amssymb}
\usepackage{amsthm}

\usepackage[matrix,arrow,curve]{xy}
\usepackage{algorithm}
\usepackage[noend]{algorithmic}

\newcommand{\field}[1][]{\mathbb{F}_{#1}}
\newcommand{\bfield}{L}
\newcommand{\subfield}{L^\sigma}

\newcommand{\Aut}[2][]{\operatorname{Aut}_{#1}(#2)}
\newcommand{\norm}[2]{N_{#1}(#2)}
\DeclareMathOperator{\rank}{rk}
\newcommand{\lclm}[1]{\left[#1\right]_\ell}

\newcommand{\gcrd}[1]{\left(#1\right)_r}
\newcommand{\matrixring}[2]{\mathcal{M}_{#1}(#2)}
\newcommand{\tovector}{\mathfrak{v}}

\usepackage{epstopdf}
\usepackage{subfigure}

\theoremstyle{plain}
\newtheorem{theorem}{Theorem}[section]

\newtheorem{lemma}[theorem]{Lemma}
\newtheorem{proposition}[theorem]{Proposition}

\theoremstyle{definition}
\newtheorem{definition}{Definition}
\newtheorem{example}{Example}

\theoremstyle{remark}
\newtheorem{remark}{Remark}

\begin{document}



\title{Peterson-Gorenstein-Zierler algorithm for skew RS codes}

\author[J. G\'omez-Torrecillas]{Jos\'e G\'omez-Torrecillas}
\email{gomezj@ugr.es}
\address{CITIC and Department of Algebra, University of Granada, Spain}
\author{F. J. Lobillo}
\email{jlobillo@ugr.es}  
\address{CITIC and Department of Algebra, University of Granada, Spain}
\author[G. Navarro]{Gabriel Navarro}
\email{gnavarro@ugr.es}
\address{CITIC and Department of Computer Sciences and AI, University of Granada, Spain}

\maketitle

\begin{abstract}
We design a non-commutative version of the Peterson-Gorenstein-Zierler decoding algorithm for a class of codes that we call skew RS codes.  These codes are left ideals of a quotient of a skew polynomial ring, which endow them of a sort of non-commutative cyclic structure. Since we work over an arbitrary field, our techniques may be applied both to linear block codes and convolutional codes. In particular, our decoding algorithm applies for block codes beyond the classical cyclic case.\end{abstract}



\section{Introduction}
From a pure mathematical perspective, a linear block code is a vector subspace $\mathcal{C}$ of $\mathbb{F}^n$, for some finite field $\mathbb{F}$.  In order to get codes with good properties, it is usual to endow both $\mathbb{F}^n$ and $\mathcal{C}$ with additional algebraic structures. This is the case of BCH codes \cite{Hocquenghem:1959, Bose/Chaudhuri:1960}, which become cyclic block codes, that is, ideals of the quotient ring $\mathbb{F}[x]/\langle x^n -1 \rangle $ of the polynomial ring $\mathbb{F}[x]$ by the ideal generated by $x^n - 1$, see for instance \cite{Huffman/Pless:2010}. Thus, they must be constructed by carefully selecting some factors of the polynomial $x^n -1$. The reward is, for instance, that these codes are built with designed Hamming distance, and that there are several efficient decoding algorithms taking advantage of their rich algebraic structure. For example, the Peterson-Gorenstein-Zierler algorithm, that makes use of linear algebra techniques; the Sugiyama algorithm, which refines some steps by means of the polynomial arithmetic; or the Sudan-Guruswami algorithm for list decoding, which can be applied to the subclass of Reed-Solomon codes.

It is observed in  \cite{Boucher/Geiselmann/Ulmer:2007, Boucher/Ulmer:2009} that the number of potential ``cyclic'' codes $\mathcal{C}$ of fixed length \(n\) is substantially increased if $\mathcal{C}$ is required to be a left ideal of a suitable quotient ring of a skew polynomial ring $\mathbb{F}[x;\sigma]$, where $\sigma$ is an automorphism of $\mathbb{F}$. It is worth to mention that, since \cite{Piret:1975, Piret:1976}, non-commutative rings are used to endow convolutional codes with non trivial cyclic structures, see also  \cite{Ross:1979, Gluesing/Schmale:2004}. Recently, in \cite{GLNSugi}, a non-commutative version of Sugiyama's decoding algorithm \cite{Sugiyama:1975} has been designed for convolutional codes built as certain left ideals of a simple ring of dimension $n$ over the fraction  field $\mathbb{F}(z)$ of the polynomial ring $\mathbb{F}[z]$, where $z$ represents the delay operator. Concretely, see \cite{GLNSCCC}, SCCC codes are defined as left ideals of the ring $\mathbb{F}(z)[x;\sigma]/\langle x^n - 1 \rangle$, where $\sigma$ is an $\mathbb{F}$--automorphism of order $n$ of the field $\mathbb{F}(z)$. 

In the present paper, a version of Peterson-Gorenstein-Zierler Algorithm for skew Reed-Solomon (RS) codes is proposed. Since we intend to cover both block and convolutional codes, we work with an abstract field $L$. Thus, we define a skew RS code as a left ideal, with generator carefully chosen, of a factor ring $L[x;\sigma]/\langle x^n -1 \rangle$, where $\sigma$ is an automorphism of $L$ of (finite) order $n$. Of course, a verbatim translation of the original Peterson-Gorenstein-Zierler Algorithm \cite{Peterson:1960, Gorenstein/Zierler:1961} makes no sense in our non-commutative context. Actually, even in the finite field block case, a skew cyclic code need not to be a cyclic code. Therefore our algorithm provides an efficient decoding procedure which applies beyond the classical Peterson-Gorenstein-Zierler algorithm. Our version shares with the block (commutative) one a general scheme where linear algebra tools and arguments play an important role, as, for instance, handling syndromes and computing errors. We also exploit the algebraic properties of the skew polynomial ring $L[x;\sigma]$ and of the field extension $L^\sigma \subset L$, which ``encode'' the skew cyclic structure. 

The paper is structured as follows. Section \ref{techprop} is devoted to state some technical results needed for proving the properties of skew RS codes and the steps of Algorithm \ref{PGZ}. In particular, we would like to highlight what we have named the Circulant Lemma (Lemma \ref{circulantlemma}) which is the key-tool that ensures the correctness of most of the methods developed in the paper. In Section \ref{RSsec} we define skew cyclic codes over a field and give a systematic procedure for generating them. Then we describe the subclass of skew RS codes and determine their Hamming distance. Section \ref{PGZsec} is devoted to state and prove a Peterson-Gorenstein-Zierler algorithm for decoding skew RS codes. Finally, in Section \ref{secexam}, we provide a selection of examples aiming to illustrate the wide range of codes to which this algorithm can be applied.

\section{Some technical results}\label{techprop}

Throughout this paper \(\bfield\) will denote a field, \(\sigma \in \Aut{L}\) a field automorphism of order \(n\), and \(\bfield^\sigma\), the invariant subfield. In this section we prove some technical results which will be used subsequently. We begin by what we call the \emph{Circulant Lemma}, which is a particular case of \cite[Corollary 4.13]{Lam/Leroy:1988}. An elementary proof can  be  found in \cite{GLNSugi}.
We recall the reader that $L$ is an $\bfield^\sigma$-vector space of dimension $n$, the order of the automorphism, see e.g. \cite[\S 4.5]{Jacobson:1985}.

\begin{lemma}[Circulant Lemma]\cite[Corollary 4.13]{Lam/Leroy:1988}\label{circulantlemma}
Let \(\{\alpha_0, \dots, \alpha_{n-1}\}\) be an $\bfield^\sigma$--basis of \(\bfield\). Then, for all \(t \leq n\) and every subset \(\{k_1, k_2, \dots, k_{t}\} \subseteq \{0, 1, \dots, n-1\}\),
\[
\begin{vmatrix}
\alpha_{k_1} & \sigma(\alpha_{k_1}) & \dots & \sigma^{t-1}(\alpha_{k_1}) \\
\alpha_{k_2} & \sigma(\alpha_{k_2}) & \dots & \sigma^{t-1}(\alpha_{k_2}) \\
\vdots & \vdots & \ddots & \vdots \\
\alpha_{k_{t}} & \sigma(\alpha_{k_{t}}) & \dots & \sigma^{t-1}(\alpha_{k_{t}})
\end{vmatrix} \neq 0.
\]
\end{lemma}

Let us introduce the following notation, which will simplify some constructions in the sequel. For each \(\gamma \in \bfield\), we denote
\[
\gamma^{[\sigma]} = 
\begin{pmatrix}
\gamma \\
\sigma(\gamma) \\
\vdots \\
\sigma^{n-1}(\gamma)
\end{pmatrix} \in \bfield^n.
\]
We also denote by \(\varepsilon_i = (0, \dots, 1, \dots, 0)\), with $0\leq i\leq n-1$, the vector with \(1\) in its \(i\)th position and \(0\), otherwise. These canonical vectors will be considered as rows or columns as required by the situation. 

For the rest of this section, let us fix an $\bfield^\sigma$--basis \(\{\alpha_0, \dots, \alpha_{n-1}\}\) of \(\bfield\). A consequence of Lemma \ref{circulantlemma} is that any different vectors \(\alpha_{k_1}^{[\sigma]}, \dots, \alpha_{k_t}^{[\sigma]}\) are linearly independent, where $\{k_1,\ldots ,k_t\}\subseteq\{0,\ldots ,n-1\}$. 
 
\begin{lemma}\label{lineardependence}
Let \(l,l_1, \dots, l_s \in \{0, \dots, n-1\}\). Then \(\alpha_l^{[\sigma]}\) is a linear combination of \(\alpha_{l_1}^{[\sigma]}, \dots, \alpha_{l_s}^{[\sigma]}\) if and only if 
$l=l_i$ for some $i\in \{1,\ldots ,s\}$.
\end{lemma}

\begin{proof} It follows trivially from Lemma \ref{circulantlemma}.
\end{proof}
We now fix a set of indices \(\{t_1, \dots, t_m\} \subseteq \{0, \dots, n-1\}\) and let
\[
A = \left(\begin{array}{c|c|c}
\alpha_{t_1}^{[\sigma]} & \dots & \alpha_{t_m}^{[\sigma]}
\end{array}\right) = 
\begin{pmatrix}
\alpha_{t_1} & \cdots & \alpha_{t_{m}} \\
\sigma(\alpha_{t_1}) & \cdots & \sigma(\alpha_{t_{m}}) \\
\vdots & \ddots & \vdots \\
\sigma^{n-1}(\alpha_{t_1}) & \cdots & \sigma^{n-1}(\alpha_{t_{m}})
\end{pmatrix}.
\]
By Lemma \ref{circulantlemma}, \(\rank A = m\), where \(\rank A\) denotes the rank of \(A\). Let \(B \in \matrixring{m \times p}{\bfield}\), a \(m \times p\) matrix over \(\bfield\), such that \(\rank B = p\) and it has no zero row. Observe that then \(p \leq m\) and \(\rank AB = p\). Let us split the set $\{1,\ldots, m\}$ into two disjoint subsets $G_1$ and $G_2$, where
\[
\rank \left( \begin{array}{c|c}
AB & \alpha_{t_i}^{[\sigma]}
\end{array} \right) = \begin{cases}
p & \text{if \(i \in G_1\),} \\
p+1 & \text{if \(i \in G_2\).}
\end{cases}
\]
Denote by $q$ the cardinal of the set $G_1$.
\begin{lemma}\label{q=p=m}
Under the above conditions and notation, \(q \leq p\). Moreover, \(q=p\) if and only if \(p=m\).
\end{lemma}

\begin{proof}
Since \(\rank AB = p\), for each $i\in G_1$, \(\alpha_{t_i}^{[\sigma]}\) is a linear combination of the columns of \(AB\). Therefore, there exists \(C \in \matrixring{p \times q}{\bfield}\) such that the matrix formed by the column vectors $\alpha_{t_i}^{[\sigma]}$, for $i\in G_1$, equals $ABC$.
Then \(\rank ABC = q\), by Lemma \ref{circulantlemma}. Consequently, \(q \leq p\). 

Assume now \(p=m\). Then \(B\) is non singular and 
\(
A = ABB^{-1}.
\)
In particular, for each \(1 \leq i \leq m\), \(\alpha_{t_i}^{[\sigma]}\) is a linear combination of the columns of \(AB\), i.e.
\[
\rank \left(\begin{array}{c|c}
AB & \alpha_{t_i}^{[\sigma]}
\end{array}\right) = m.
\]
In other words, \(G_1=\{1,\ldots, m\}\), so \(q=m=p\). 

Finally, suppose \(q=p\) and assume \(q < m\). We may reorder the columns of $A$ in such a way that the first $q$ columns correspond to the indices in $G_1$, that is, $G_1=\{1,\ldots, q\}$ and $G_2=\{q+1,\ldots ,m\}$. Let $P$ be the permutation matrix which provides this reordering of the columns of $A$, hence $AB=APP^{-1}B$. Observe that $P^{-1}B$ also has rank $p$ and all its rows are non zero. So, for simplicity, we abuse notation and denote also by $A$ the matrix $AP$, and by $B$, the matrix $P^{-1}B$. We divide the matrices as
\[
A = \left(\begin{array}{c|c} A_0 & A_1 \end{array}\right) \text{ and }  B = \left(\begin{array}{c} B_0 \\ \hline B_1 \end{array}\right),
\]
where \(A_0\) encompasses the first \(q\) columns of $A$ and \(B_0\), the first \(q\) rows of $B$. 
Then \[
q = \rank \left(\begin{array}{c|c} AB & A_0 \end{array}\right) = \rank \left(\begin{array}{c|c} A_0 B_0 + A_1 B_1 & A_0 \end{array}\right) = \rank \left(\begin{array}{c|c} A_1 B_1 & A_0 \end{array}\right).
\]
By hypothesis, the last row of \(B_1\) is non zero. Let \((b_{q+1}, \dots, b_m)\) be a column of \(B_1\) with \(b_m \neq 0\). Then 
\[
A_1 \left(\begin{smallmatrix} b_{q+1} \\ \vdots \\ b_m \end{smallmatrix}\right) = \alpha_{t_{q+1}}^{[\sigma]} b_{q+1} + \dots + \alpha_{t_{m}}^{[\sigma]} b_{m}
\] 
is a linear combination of the columns of \(A_0\), i.e. there exist \(c_1, \dots, c_q \in \bfield\) such that 
\[
\sum_{i=1}^q \alpha_{t_i}^{[\sigma]} c_i = \sum_{i=q+1}^m \alpha_{t_i}^{[\sigma]} b_i
\]
and, consequently, \(\alpha_{t_1}^{[\sigma]}, \dots, \alpha_{t_m}^{[\sigma]}\) are linearly dependent, contradicting Lemma \ref{circulantlemma}. Thus \(p=m\).
\end{proof}

We shall also need the following straightforward technical lemma.

\begin{lemma}\label{tech}
Let \(H \in \matrixring{m \times n}{\bfield}\) be a matrix in reduced row echelon form and let \(\varepsilon_i\) be a canonical vector of length \(n\). Then
\[
\rank \left( \begin{array}{c} H \\ \hline \varepsilon_i \end{array} \right) = \rank H
\]
if and only if \(\varepsilon_i\) is a row of \(H\).
\end{lemma}

\section{Skew cyclic codes}\label{RSsec}

In this section we introduce skew cyclic codes over $L$ and give some of their properties. Let us denote by $R$ the skew polynomial ring $L[x;\sigma]$, see \cite{Ore:1933}. Given $f,g\in R$, we say that $f$ right divides $g$, $f\mid_r g$, if $Rg\subseteq Rf$, i.e. the remainder of the left division of $g$ by $f$ is zero. For $\gamma \in L$, the right evaluation of $f$ in $\gamma$ is the remainder of the left division of $f$ by $x-\gamma$. 
As \(R\) is a left (and right) PID, there are least common multiples and greatest common divisors of polynomials on both sides. We use the notation
\[
\lclm{f,g} \text{ and } \gcrd{f,g}
\]
to refer to the least common left multiple and the greatest common right divisor of a pair \(f,g \in R\). Then,
\[
Rf + Rg = R\gcrd{f,g} \text{ and } Rf \cap Rg = R\lclm{f,g}.
\]

Since we are assuming that the order of \(\sigma\) is \(n\), the polynomial \(x^n-1\) is central in $R$, so we may consider the quotient ring $\mathcal{R}= \bfield[x;\sigma] / \langle x^n-1\rangle$. Throughout, we shall see the elements in $\mathcal{R}$ as polynomials of degree lower than $n$. As an $L$-vector space, $\mathcal{R}$ is isomorphic to $L^n$ via the coordinate map \(\tovector:\mathcal{R} \to \bfield^n\), mapping each polynomial (of degree lower than \(n\)) to the vector formed by its coefficients. We shall use freely this identification all along the paper. 

Linear codes over $L$ are vector subspaces of \(\bfield^n\). When \(\bfield = \field[q]\), a finite field, these are block codes, whilst, when \(\bfield = \field[q](z)\), these are convolutional codes \cite{Forney:1970}. So we adopt the coding theory terminology over $\bfield$. Following this philosophy, we may define the \emph{weight} of a vector and the \emph{Hamming distance} of a code as usual. All these notions are going to be used without further mention in the framework of a vector space over \(\bfield\).

\begin{definition}
A \emph{skew cyclic code} over \(\bfield\) is a vector subspace \(\mathcal{C} \leq \bfield^n\) such that \(\tovector^{-1}(\mathcal{C})\) is a left ideal of \(\mathcal{R}\). Equivalently, it is a vector subspace \(\mathcal{C} \leq \bfield^n\) such that
\[
(a_0, \dots, a_{n-2}, a_{n-1}) \in \mathcal{C} \Rightarrow (\sigma(a_{n-1}), \sigma(a_0), \dots, \sigma(a_{n-2})) \in \mathcal{C}.
\]
\end{definition}

These codes have been introduced under the name of \(\sigma\)--cyclic codes in \cite{Boucher/Geiselmann/Ulmer:2007}, when $\bfield$ is a finite field, and skew cyclic convolutional codes in \cite{GLNSCCC}, when \(\bfield = \field[q](z)\). 

Any left ideal of \(\mathcal{R}\) is principal, so every skew cyclic code is generated by a polynomial in \(R\). Obviously, the generator can be taken as a right divisor of \(x^n - 1\), so it is important to determine the decompositions of $x^n-1$ in $R$. First we recall some notions about skew polynomials. Let \(\gamma \in \bfield\), the \(i\)th-norm of \(\gamma\) is defined to be
\[
\norm{i}{\gamma} = \gamma \sigma(\gamma) \dots \sigma^{i-1}(\gamma).
\]
Norms are useful to evaluate skew polynomials. In fact, if \(f = \sum_{i \geq 0} f_i x^i \in R\), the remainder of the left division of \(f\) by \(x-\gamma\) is 
\begin{equation}\label{polyevaluation}
\textstyle\sum_{i \geq 0}f_i \norm{i}{\gamma}.
\end{equation}
This is an easy computation that can be found in \cite{Leroy:1995} in a more general context. We also recall some formulas in order to ease the reading of the paper. Concretely, if \(\alpha,\beta,\gamma \in \bfield\) such that \(\beta = \alpha^{-1}\sigma(\alpha)\), then
\begin{equation}\label{normproperties}
\begin{split}
\norm{i}{\sigma^k(\gamma)} &= \sigma^k(\norm{i}{\gamma}), \\
\norm{i}{\sigma^k(\beta)} &= \sigma^k(\alpha)^{-1} \sigma^{k+i}(\alpha).
\end{split}
\end{equation}

We shall follow the systematic method described in \cite{GLNSCCC} in order to construct skew cyclic codes. By the Normal Basis Theorem, we may choose an element $\alpha\in \bfield$ such that $\{\alpha, \sigma(\alpha),\ldots ,\sigma^{n-1}(\alpha)\}$ is a basis of $\bfield$ as an $\bfield^\sigma$-vector space. Set now \(\beta = \alpha^{-1}\sigma(\alpha)\). We shall fix this notation for the rest of the paper. 

\begin{lemma}\label{degreelclm}
For any subset \(T=\{t_1,t_2,\cdots, t_m\}\subseteq \{0,1,\ldots ,n-1\}\), the polynomial 
\[
g=\lclm{x-\sigma^{t_1}(\beta),x-\sigma^{t_2}(\beta),\ldots ,x-\sigma^{t_m}(\beta)}
\] 
has degree $m$. Consequently, if \(x-\sigma^{s}(\beta) \mid_{r} g\), then \(s\in T\).
\end{lemma}
\begin{proof}
The statement can be proved following the same steps of the proof of \cite[Lemma 2]{GLNSugi}
\end{proof}

As a consequence of Lemma \ref{degreelclm}, 
\[
x^n - 1 = \lclm{x-\beta, x-\sigma(\beta), \dots, x-\sigma^{n-1}(\beta)},
\]
since, by \eqref{normproperties}, $\norm{n}{\sigma^k(\beta)} = \sigma^k(\alpha)^{-1}\sigma^{k+n}(\alpha) = \sigma^k(\alpha^{-1}\alpha) = 1, $
and then \(x-\sigma^k(\beta)\) right divides \(x^n-1\) for all \(0 \leq k \leq n-1\). Therefore, given $\{t_1,\ldots, t_k\}\subseteq \{0,1,\ldots ,n-1\}$, the polynomial $g=\lclm{x-\sigma^{t_1}(\beta),\ldots , x-\sigma^{t_k}(\beta)}$ generates a left ideal $\mathcal{R}g$ such that $\mathfrak{v}(\mathcal{R}g)$ is a skew cyclic code of dimension $n-k$.

For the convenience, we call $\beta$-\emph{roots} to the elements of the set  \(\{\beta, \sigma(\beta), \dots, \sigma^{n-1}(\beta)\}\). By \eqref{polyevaluation} and \eqref{normproperties}, given a polynomial \(f = \sum_{i=0}^{n-1} f_i x^i\in \mathcal{R}\),
\begin{equation}\label{findingroots}
\textstyle x - \sigma^j(\beta) \mid_r f \iff \sum_{i=0}^{n-1} f_i \norm{i}{\sigma^j(\beta)} = 0 \iff \sum_{i=0}^{n-1} f_i \sigma^{i+j}(\alpha) = 0.
\end{equation}
Let then $N$ be the matrix formed by the norms of the $\beta$-roots,
\begin{equation}\label{evaluationmatrix}
N = \begin{pmatrix}
\norm{0}{\beta} & \norm{0}{\sigma(\beta)} & \cdots & \norm{0}{\sigma^{n-1}(\beta)} \\
\norm{1}{\beta} & \norm{1}{\sigma(\beta)} & \cdots & \norm{1}{\sigma^{n-1}(\beta)} \\
\vdots & \vdots & \ddots & \vdots \\
\norm{n-1}{\beta} & \norm{n-1}{\sigma(\beta)} & \cdots & \norm{n-1}{\sigma^{n-1}(\beta)}
\end{pmatrix},
\end{equation}
the components of \(\mathfrak{v}(f) N=(f_0, \dots, f_{n-1})  N\) are the right evaluations of \(f\) in the set of $\beta$-roots, i.e. the vector formed by the left remainders of $f$ by the polynomials $x-\sigma^i(\beta)$ for $i=0,\ldots ,n-1$. Hence, the diagram   
\[
\xymatrix{ \mathcal{R} \ar@/^5pt/[rd]^-{ev} \ar[d]^-{\mathfrak{v}} & \\
\bfield^n \ar[r]^-{\cdot N} & \bfield^n,}
\]
is a commutative diagram of \(\bfield\)-linear isomorphisms, where $ev$ maps each polynomial $f$ to the $n$-tuple formed by the left remainders of $f$ by $x-\sigma^i(\beta)$ for $i=0,\ldots ,n-1$. Indeed, by Lemma \ref{circulantlemma}, $N$ is non singular, so it provides a change of basis. We call the \emph{set of $\beta$-roots of \(f\)} to the set formed by the $\beta$-roots $\gamma$ verifying that $x-\gamma \mid_r f$, that is by those corresponding to the zero coordinates of \((f_0, \dots, f_{n-1}) N\).

We say that, a non-constant right divisor \(f \mid_r x^n-1\), \emph{fully $\beta$-decomposes} if there exists \(\{t_1, \dots, t_m\} \subseteq \{0,1,\dots,n-1\}\) such that 
\[
f = \lclm{x-\sigma^{t_1}(\beta), \dots, x-\sigma^{t_m}(\beta)}.
\]
Observe that, by Lemma \ref{degreelclm}, \(\deg f = m\), the cardinal of the set of $\beta$-roots of $f$.

\begin{lemma}\label{rreffullydecomposable}
Let \(f = \sum_{i=0}^{m} f_i x^i \in \mathcal{R}\) with \(f_m \neq 0\) and
\[
M_f = \begin{pmatrix} 
f_0 & f_1 & \ldots & f_m & 0 & \ldots & 0 \\
0 & \sigma(f_0) & \ldots & \sigma(f_{m-1}) & \sigma(f_{m}) & \ldots & 0\\
&  & \ddots &  &  & \ddots & \\
0 & \ldots & 0 & \sigma^{n-m-1}(f_0)& \ldots & \ldots  &\sigma^{n-m-1}(f_m)
\end{pmatrix}_{(n-m)\times n}.
\]
Then the rows of \(M_f\) are a basis of \(\mathfrak{v}(\mathcal{R}f)\) as an \(\bfield\)--vector space. Moreover, \(f\) fully $\beta$-decomposes if and only if 
\[
\operatorname{rref}(M_f N) = \left( \begin{array}{c}
\varepsilon_{i_1} \\
\hline 
\vdots \\ 
\hline
\varepsilon_{i_{n-m}}
\end{array}\right)
\]
for some \(0 \leq i_1 < \dots < i_{n-m} \leq n-1\), where \(\operatorname{rref}\) denotes the reduced row echelon form.
\end{lemma}

\begin{proof}
An \(\bfield\)-basis of \(\mathcal{R}f\) is \(\{f, xf, \dots, x^{n-m-1}f\}\), whose coordinates correspond to the rows of \(M_f\).  Now, \(f = \lclm{x-\sigma^{t_1}(\beta), \dots, x-\sigma^{t_m}(\beta)}\) if and only if any left multiple of \(f\) is also a left multiple of \(x-\sigma^{t_i}(\beta)\) for \(1 \leq i \leq m\), if and only if the $t_i$-th columns of \(M_f N\) are zero for $i=1,\ldots ,m$. Since \(M_f N\) has \(n-m\) rows, rank $n-m$ and \(n-m\) non zero columns, the result follows.
\end{proof}

We shall need the following result. 

\begin{lemma}\label{fullydecomposable}
Let \(f,g \in \mathcal{R}\) be fully $\beta$-decomposable polynomials. Then \(\gcrd{f,g}\) and \(\lclm{f,g}\) are also fully $\beta$-decomposable. 
\end{lemma}
\begin{proof}
Since \(f,g\) are fully $\beta$-decomposable, there exist subsets \(T_1,T_2 \subseteq \{0, \dots, n-1\}\) such that 
\[
f = \lclm{\{x-\sigma^{i}(\beta)\}_{i\in T_1}}  \text{ and } g =  \lclm{\{x-\sigma^{i}(\beta)\}_{i\in T_2}}.
\]
Hence, it is straightforward that
\[
\lclm{f,g} = \lclm{\{x-\sigma^{i}(\beta)\}_{i\in T_1\cup T_2}}.
\]
On the other hand, 
\[
\lclm{\{x-\sigma^{i}(\beta)\}_{i\in T_1\cap T_2}} \mid_r \gcrd{f,g}.
\]
The formula \(\deg f + \deg g = \deg \gcrd{f,g} + \deg \lclm{f,g}\) and Lemma \ref{degreelclm} give the equality.
\end{proof}

We now may define the class of skew cyclic codes for which the decoding algorithm of the next section can be applied.

\begin{definition}\label{RS}
Under the conditions and notation of this section, a \emph{skew Reed-Solomon (RS) code} of designed Hamming distance $\delta$ is a skew cyclic code $\mathcal{C}$ such that $\mathfrak{v}^{-1}(\mathcal{C})$ is generated by a polynomial $\lclm{x-\sigma^r(\beta), x-\sigma^{r+1}(\beta), \ldots , x-\sigma^{r+\delta-2}(\beta)}$ for some $r\geq 0$.
\end{definition}

\begin{theorem}\label{MDScode}
A skew RS code of designed Hamming distance $\delta$ has Hamming distance $\delta$. Consequently, it is an MDS code. 
\end{theorem}

\begin{proof}
The result can be proved analogously to the proof of \cite[Theorem 4]{GLNSugi}.
\end{proof}

\section{A Peterson-Gorenstein-Zierler decoding algorithm}\label{PGZsec}

Throughout this section $\mathcal{C}$ denotes a skew RS code as described in Definition \ref{RS}. Without loss of generality, we may assume that $\mathcal{C}$ is a narrow-sense skew RS code, i.e. we may set $r=0$. This is because we always may write \(\alpha' = \sigma^r(\alpha)\), which also provides a normal basis, and then $\sigma^r(\beta) = \beta'=(\alpha')^{-1}\sigma(\alpha')$, so $\lclm{x-\beta',\ldots ,x-\sigma^{\delta-2}(\beta')}$ is a generator of $\mathcal{C}$. Therefore,  suppose that  the left ideal $\mathfrak{v}^{-1}(\mathcal{C})$ is generated by \(g=\lclm{x-\beta, x-\sigma(\beta), \dots , x-\sigma^{\delta-2}(\beta)}\) for some $2\leq \delta \leq n$. By Theorem \ref{MDScode}, the Hamming distance of $\mathcal{C}$ is exactly $\delta$ and, following Algorithm \ref{PGZ} below, it can correct up to $t=\left\lfloor \frac{\delta-1}{2} \right\rfloor$ errors, the error correction capability of the code.

Suppose now that a message $m=(m_0,\ldots ,m_{n-\delta})$  must be transmitted through a noisy channel. Since we are identifying $\mathcal{C}$ with $\mathfrak{v}^{-1}(\mathcal{C})$, $m=\sum_{i=0}^{n-\delta}m_ix^i$. The message $m$ is encoded to a codeword  \(c=mg\) and \(y = c + e\) is received, where \(e = e_1 x^{k_1} + \dots + e_\nu x^{k_\nu}\), with \(\nu \leq t\), is the error polynomial.  The purpose of this section is to develop an algorithm, following the scheme of the classical Peterson-Gorenstein-Zierler decoding algorithm, for computing this error.

For each \(0 \leq i \leq 2t-1 \), the \(i\)th \emph{syndrome} $s_i$ of the received polynomial $y$ is defined to be the left remainder of \(y\) by \(x-\sigma^i(\beta)\). Since \(c\) is right divisible by \(x-\sigma^i(\beta)\) for $i=0,\ldots ,\delta-2$; it follows, by \eqref{normproperties}, that
\begin{equation}\label{eq1}
\begin{split}
s_i &= \sum_{j=0}^{n-1} y_j \norm{j}{\sigma^i(\beta)} = \sum_{j=1}^\nu e_j \norm{k_j}{\sigma^i(\beta)}  \\ &
= \sum_{j=1}^\nu e_j\sigma^i(\alpha^{-1})\sigma^{i+k_j}(\alpha) = \sigma^i(\alpha^{-1})\sum_{j=1}^\nu e_j\sigma^{i+k_j}(\alpha).
\end{split}
\end{equation}

\begin{proposition}\label{errorvalues}
The error values \((e_1, \dots, e_\nu)\) are the unique solution of the linear system
$$X\,
\begin{pmatrix}
\sigma^{k_1}(\alpha) & \sigma^{k_1+1}(\alpha) & \cdots & \sigma^{k_1+\nu-1}(\alpha) \\
\sigma^{k_2}(\alpha) & \sigma^{k_2+1}(\alpha) & \cdots & \sigma^{k_2+\nu-1}(\alpha) \\
\vdots & \vdots & \ddots  & \vdots \\
\sigma^{k_\nu}(\alpha) & \sigma^{k_\nu+1}(\alpha) & \cdots & \sigma^{k_\nu+\nu-1}(\alpha)\\
\end{pmatrix} 
=
(\alpha s_0, \sigma(\alpha) s_1, \cdots, \sigma^{\nu-1}(\alpha) s_{\nu-1}).$$
\end{proposition}

\begin{proof}
Observe that
\[
\begin{pmatrix}
\sigma^{k_1}(\alpha) & \sigma^{k_1+1}(\alpha) & \cdots & \sigma^{k_1+\nu-1}(\alpha) \\
\sigma^{k_2}(\alpha) & \sigma^{k_2+1}(\alpha) & \cdots & \sigma^{k_2+\nu-1}(\alpha) \\
\vdots & \vdots & \ddots & \vdots \\
\sigma^{k_\nu}(\alpha) & \sigma^{k_\nu+1}(\alpha) & \cdots & \sigma^{k_\nu+\nu-1}(\alpha)\\
\end{pmatrix} 
=
\begin{pmatrix}
\sigma^{k_1}(\alpha) & \sigma(\sigma^{k_1}(\alpha)) & \cdots & \sigma^{\nu-1}(\sigma^{k_1}(\alpha)) \\
\sigma^{k_2}(\alpha) & \sigma(\sigma^{k_2}(\alpha)) & \cdots & \sigma^{\nu-1}(\sigma^{k_2}(\alpha)) \\
\vdots & \vdots & \ddots & \vdots \\
\sigma^{k_\nu}(\alpha) & \sigma(\sigma^{k_\nu}(\alpha)) & \cdots & \sigma^{\nu-1}(\sigma^{k_\nu}(\alpha))\\
\end{pmatrix},
\]
which has non zero determinant, by Lemma \ref{circulantlemma}. By \eqref{eq1}
\[
\sigma^i(\alpha)s_i=\sum_{j=1}^\nu e_j\sigma^{i+k_j}(\alpha),
\]
so the result follows.
\end{proof}

By Proposition \ref{errorvalues}, the decoding process is therefore reduced to find the error positions \(\{k_1, \dots, k_\nu\}\). We define the \emph{error locator} polynomial as 
\[
\lambda = \lclm{x-\sigma^{k_1}(\beta), x-\sigma^{k_2}(\beta), \ldots , x-\sigma^{k_\nu}(\beta)}.
\] 
By Lemma \ref{degreelclm}, $\lambda$ has degree $\nu$, and, once $\lambda$ is known, the error positions can be determined. 

Recall that $f=\sum_{k=0}^{n-1} f_k x^k\in \mathcal{R}\lambda$ if and only if $x-\sigma^{k_j}(\beta)\mid_r f$ for all $j=1,\ldots ,\nu$; or, equivalently, $\sum_{k=0}^{n-1} f_k \norm{k}{\sigma^{k_j}(\beta)}=0$ for all $j=1,\ldots , \nu$. Therefore, $(f_0,\ldots, f_{n-1})\in \mathfrak{v}(\mathcal{R}\lambda)$ if and only if $f$ satisfies the equation $(f_0,\ldots, f_{n-1}) T=0$, where
\[
T=\begin{pmatrix}
\norm{0}{\sigma^{k_1}(\beta)} & \norm{0}{\sigma^{k_2}(\beta)} & \cdots & \norm{0}{\sigma^{k_\nu}(\beta)} \\
\norm{1}{\sigma^{k_1}(\beta)} & \norm{1}{\sigma^{k_2}(\beta)} & \cdots & \norm{1}{\sigma^{k_\nu}(\beta)} \\
\vdots & \vdots &  & \vdots \\
\norm{n-1}{\sigma^{k_1}(\beta)} & \norm{n-1}{\sigma^{k_2}(\beta)} & \cdots & \norm{n-1}{\sigma^{k_\nu}(\beta)} \\
\end{pmatrix}.
\]
Now, by \eqref{normproperties}, $\norm{k}{\sigma^{k_j}(\beta)}=\sigma^{k_j}(\alpha)^{-1}\sigma^{k_j+k}(\alpha)$, then $\mathcal{R}\lambda$ corresponds to the left kernel of the matrix
\[
\Sigma = \begin{pmatrix}
\sigma^{k_1}(\alpha) & \sigma^{k_2}(\alpha) & \cdots & \sigma^{k_\nu}(\alpha) \\
\sigma^{k_1+1}(\alpha) & \sigma^{k_2+1}(\alpha) & \cdots & \sigma^{k_\nu+1}(\alpha) \\
\vdots & \vdots & \ddots & \vdots \\
\sigma^{k_1+n-1}(\alpha) & \sigma^{k_2+n-1}(\alpha) & \cdots & \sigma^{k_\nu+n-1}(\alpha) \\
\end{pmatrix} = \left( \begin{array}{c} \Sigma_0 \\ \hline \Sigma_1 \end{array} \right),
\]
where $\Sigma_0$ comprises the first $\nu+1$ rows of $\Sigma$. Let us consider the matrix
\[
E=\begin{pmatrix} e_1 & \sigma^{-1}(e_1) & \cdots & \sigma^{-\nu+1}(e_1) \\
e_2 & \sigma^{-1}(e_2) & \ldots & \sigma^{-\nu+1}(e_2) \\
\vdots & \vdots & \ddots & \vdots \\
e_\nu & \sigma^{-1}(e_\nu) & \cdots & \sigma^{-\nu+1}(e_\nu)
\end{pmatrix}.
\]
Thus $S=\Sigma E$ is an $(n \times \nu)$-matrix whose $(k,i)$-component is given by $\sum_{j=1}^{\nu}\sigma^{-i}(e_j)\sigma^{k_j+k}(\alpha)$. Observe that, by (\ref{eq1}), whenever $k+i<2t-1$, this component can be written as $\sigma^{-i}(s_{k+i})\sigma^k(\alpha)$, so that we may divide $S = \left ( \begin{array}{c} S_0 \\ \hline S_1 \end{array} \right )$, where
\[
S_0 = \begin{pmatrix} s_0\alpha & \sigma^{-1}(s_1)\alpha & \ldots & \sigma^{-\nu+1}(s_{\nu-1})\alpha \\
s_1\sigma(\alpha) & \sigma^{-1}(s_2)\sigma(\alpha) & \ldots & \sigma^{-\nu+1}(s_\nu)\sigma(\alpha) \\
\vdots & \vdots & & \vdots \\
s_\nu\sigma^\nu(\alpha) & \sigma^{-1}(s_{\nu+1})\sigma^\nu(\alpha) & \ldots & \sigma^{-\nu+1}(s_{2\nu-1})\sigma^\nu(\alpha) \\
\end{pmatrix}_{(\nu + 1) \times \nu},
\]
whose coefficients can be computed from the received polynomial $y$. In order to calculate the parameter $\nu$, the number of error positions, we may follow the scheme of Peterson-Gorenstein-Zierler algorithm for BCH block codes, see e.g. \cite[\S 5.4]{Huffman/Pless:2010}. For any $1\leq r\leq t$, let us denote by $S^r$ the matrix
\[
S^r = \begin{pmatrix} s_0\alpha & \sigma^{-1}(s_1)\alpha & \ldots & \sigma^{-r+1}(s_{r-1})\alpha \\
s_1\sigma(\alpha) & \sigma^{-1}(s_2)\sigma(\alpha) & \ldots & \sigma^{-r+1}(s_r)\sigma(\alpha) \\
\vdots & \vdots & \ddots & \vdots \\
s_t\sigma^t(\alpha) & \sigma^{-1}(s_{t+1})\sigma^t(\alpha) & \ldots & \sigma^{-r+1}(s_{t+r-1})\sigma^t(\alpha) \\
\end{pmatrix}_{(t + 1) \times r} .\]
Observe that, for all \(r \leq t\), \(S^r = \Sigma^t E^r\), where 
\[
E^r=\begin{pmatrix} e_1 & \sigma^{-1}(e_1) & \cdots & \sigma^{-r+1}(e_1) \\
e_2 & \sigma^{-1}(e_2) & \ldots & \sigma^{-r+1}(e_2) \\
\vdots & \vdots & \ddots & \vdots \\
e_\nu & \sigma^{-1}(e_\nu) & \cdots & \sigma^{-r+1}(e_\nu)
\end{pmatrix}_{\nu \times r}
\]
and 
\[
\Sigma^t = \begin{pmatrix}
\sigma^{k_1}(\alpha) & \sigma^{k_2}(\alpha) & \cdots & \sigma^{k_\nu}(\alpha) \\
\sigma^{k_1+1}(\alpha) & \sigma^{k_2+1}(\alpha) & \cdots & \sigma^{k_\nu+1}(\alpha) \\
\vdots & \vdots & \ddots & \vdots \\
\sigma^{k_1+t}(\alpha) & \sigma^{k_2+t}(\alpha) & \cdots & \sigma^{k_\nu+t}(\alpha) \\
\end{pmatrix}_{(t+1) \times \nu}.
\]

\begin{lemma}\label{ranksrelated}
For each \(r \leq t\), \(\rank S^r = \rank \Sigma E^r = \rank E^r\).
\end{lemma}

\begin{proof}
By Lemma \ref{circulantlemma}, \(\rank \Sigma = \rank \Sigma^t = \nu\). Using Sylvester's rank inequality,
\[
\min \{\rank \Sigma, \rank E^r\} \geq \rank \Sigma E^r \geq \rank \Sigma + \rank E^r - \nu = \rank E^r.
\]
Then \(\rank \Sigma E^r = \rank E^r\). Analogously, \(\rank S^r = \rank E^r\). 
\end{proof}

Hence, we may calculate the greatest $r$ such that $S^r$ has full rank. By Lemma \ref{ranksrelated}, it is also the greatest integer \(r \leq t\) such that \(E^r\) and \(\Sigma E^r\) have full rank. We shall denote by $\mu$ such a maximum.

\begin{lemma}\label{rankES}
For each \(\mu \leq r \leq t\), \(\mu = \rank E^r = \rank S^r\). Consequently, \(\mu \leq \nu\). 
\end{lemma}

\begin{proof}
By Lemma \ref{ranksrelated}, \(\mu = \rank E^\mu = \rank S^\mu\), so assume \(\mu < r\). By maximality of \(\mu\), the \((\mu+1)\)th column of \(E^r\) is a linear combination of the \(\mu\) preceding columns. Applying \(\sigma^{-1}\) we get that the \((\mu + 2)\)th column is a linear combination of the columns at positions from the second to the \((\mu+1)\)th, and hence a linear combination of the first \(\mu\) columns. Repeating the process we obtain that all columns from the \((\mu+1)\)th to \(r\)th are linear combinations of the first \(\mu\) columns, which implies that \(\rank E^r = \mu\). Since \(E^r\) has \(\nu\) rows, \(\mu \leq \nu\). Finally, \(\rank S^r = \mu\) again by Lemma \ref{ranksrelated}.
\end{proof}

\begin{proposition}\label{Vskewcode}
The left kernel \(V\) of the matrix \(\Sigma E^\mu\) is a skew cyclic code. Consequently, \(\tovector^{-1}(V) = \mathcal{R}\rho\) for some polynomial \(\rho \in \mathcal{R}\) of degree \(\mu\). Moreover, \(\rho\) is a right divisor of \(\lambda\). 
\end{proposition}

\begin{proof}
The first statement is ensured by proving that, if $(a_0, \dots, a_{n-2}, a_{n-1})\in V \subseteq L^n$, then we have $(\sigma(a_{n-1}), \sigma(a_0), \dots, \sigma(a_{n-2})) \in V$. This is due to the fact that $xa_ix^i=\sigma(a_i)x^{i+1}$ for every $0\leq i\leq n-1$. Suppose then $(a_0,a_1,\ldots, a_{n-1}) \Sigma E^{\mu} = 0$. The maximality of \(\mu\) ensures that the last column of $E^{\mu+1}$ is a linear combination of the former $\mu$ columns. Hence $(a_0,a_1,\ldots, a_{n-1}) \Sigma E^{\mu+1}=0$. Therefore,
\[
\begin{split}
0 &= (a_0,a_1,\ldots, a_{n-1}) \Sigma E^{\mu+1} \\
&= (a_0,a_1,\ldots, a_{n-1}) \left(\begin{array}{c|c} 0 & I_{n-1}\\ \hline1 & 0 \end{array} \right) \left(\begin{array}{c|c} 0 & 1\\ \hline I_{n-1} & 0\end{array} \right) \Sigma E^{\mu+1} \\
&= (a_{n-1},a_0,\ldots, a_{n-2})  \left(\begin{array}{c|c} 0 & 1\\ \hline I_{n-1} & 0\end{array} \right) \Sigma E^{\mu+1}.
\end{split}
\]
Applying $\sigma$ to this matrix equation (componentwise),
\[
(\sigma(a_{n-1}),\sigma(a_0),\ldots, \sigma(a_{n-2})) \left(\begin{array}{c|c} 0 & 1\\ \hline I_{n-1} & 0\end{array} \right) \sigma(\Sigma) \sigma(E^{\mu+1}) = 0.
\]
Observe that $\Sigma = \left(\begin{array}{c|c} 0 & 1\\ \hline I_{n-1} & 0\end{array} \right) \sigma(\Sigma)$ and \(\sigma(E^{\mu+1}) = \left(\begin{array}{c|c} \begin{smallmatrix} \sigma(e_1) \\ \vdots \\ \sigma(e_{\nu}) \end{smallmatrix} & E^{\mu}\end{array}\right)\), so 
\[
(\sigma(a_{n-1}),\sigma(a_0),\ldots, \sigma(a_{n-2})) \Sigma \left(\begin{array}{c|c} \begin{smallmatrix} \sigma(e_1) \\ \vdots \\ \sigma(e_{\nu}) \end{smallmatrix} & E^{\mu}\end{array}\right) = 0.
\]
In particular, $(\sigma(a_{n-1}),\sigma(a_0),\ldots, \sigma(a_{n-2})) \Sigma E^{\mu} = 0$, so $(\sigma(a_{n-1}),\sigma(a_0),\ldots, \sigma(a_{n-2})) \in V$, as desired. Now, any left ideal of $\mathcal{R}$ is principal, so $\mathfrak{v}^{-1}(V)$ is generated by a polynomial $\rho\in\mathcal{R}$. Since $\mathfrak{v}(\mathcal{R}\lambda)$ is the left kernel of the matrix $\Sigma$, it follows that \(\mathcal{R}\lambda \subseteq \mathcal{R}\rho\), hence \(\rho\) right divides \(\lambda\). Finally, the dimension of \(\mathcal{R}\rho\) as an \(\bfield\)-vector space is \(n - \deg \rho\). By Lemma \ref{ranksrelated}, \(\rank \Sigma E^\mu = \mu\), so \(\deg \rho = \mu\). 
\end{proof}

\begin{lemma}\label{rhocomputation}
The reduced column echelon form of \(S^t\) is 
\[
\operatorname{rcef}(S^t) = \left(\begin{array}{c|c} \begin{array}{c} I_\mu \\ \hline a_0 \cdots a_{\mu-1} \\ \hline H' \end{array} & 0_{(t+1)\times(t-\mu)} \end{array}\right),
\]
where \(I_\mu\) is the \(\mu \times \mu\) identity matrix and \(a_0, \dots, a_{\mu-1} \in \bfield\) such that \(\rho = x^\mu - \sum_{i=0}^{\mu-1} a_i x^i\).
\end{lemma}

\begin{proof}
By Lemma \ref{rankES}, \(\rank S^t = \mu = \rank S^\mu\), so
\[
\operatorname{rcef}(S^t) = \left( \begin{array}{c|c} \operatorname{rcef}(S^\mu) & 0_{(t+1)\times(t-\mu)} \end{array} \right).
\] 
Recall that \(S^\mu\) consists of the first \(t+1\) rows of \(\Sigma E^\mu\) and both have the same rank \(\mu\), therefore \(\operatorname{rcef}(S^\mu)\) is composed by the first \(t+1\) rows of \(\operatorname{rcef}(\Sigma E^\mu)\). By Proposition \ref{Vskewcode}, \(\tovector(\mathcal{R}\rho)\) is the left kernel of the matrix \(\operatorname{rcef}(\Sigma E^\mu)\). A non zero solution of the homogeneous system
\begin{equation}\label{rhoeq1}
X \left( \begin{array}{c|c}
\operatorname{rcef}(\Sigma E^\mu) &
\begin{array}{c}
0 \\ \hline I_{n-(\mu+1)}
\end{array}
\end{array}\right) = 0
\end{equation}
is a non zero element of \(\tovector(\mathcal{R}\rho)\) whose $n-(\mu+1)$ last coordinates are zero. Since \(\rho\) has degree \(\mu\), and its degree is minimal in \(\mathcal{R}\rho\), it follows that \(\tovector(\rho)\) is the unique solution, up to scalar multiplication, of \eqref{rhoeq1}. Let \(S^\mu_0\) be formed by the first \(\mu+1\) rows of $S^\mu$. Then 
\[
\operatorname{rcef}(S^\mu) = \left( \begin{array}{c} \operatorname{rcef}(S^\mu_0) \\ \hline H' \end{array} \right).
\]
Further column reductions using the identity matrix in the right block of the matrix in \eqref{rhoeq1}, allow us to see that \(\rho\) is also the non zero solution, up to scalar multiplication, of the homogeneous system 
\begin{equation}\label{rhoeq2}
X \left( \begin{array}{c|c}
\operatorname{rcef}(S^\mu_0) & 0 \\
\hline
0 & I_{n-(\mu+1)}
\end{array}\right) = 0.
\end{equation}
The size of \(\operatorname{rcef}(S^\mu_0)\) is \((\mu+1)\times\mu\). Moreover \(\rank \operatorname{rcef}(S^\mu_0) = \mu\) because the space of solutions of \eqref{rhoeq2} has dimension \(1\). So there is only one row of \(\operatorname{rcef}(S^\mu_0)\) without a pivot. If this row is not the last row then there would be a non zero polynomial in \(\mathcal{R}\rho\) of degree strictly below \(\mu\), which is impossible. Hence 
\[
\operatorname{rcef}(S^\mu_0) = \left(\begin{array}{c} I_\mu \\ \hline a_0 \cdots a_{\mu-1} \end{array}\right).
\]
Finally, \((-a_0, \dots, -a_{\mu-1},1,0, \dots, 0)\) is a non zero solution of \eqref{rhoeq2}, it follows \(\rho = x^\mu - \sum_{i=0}^{\mu-1} a_i x^i\).
\end{proof}

\begin{lemma}\label{Rrhocomputation}
If the left ideal \(\mathcal{R}\rho\) corresponds, via $\mathfrak{v}$, to the left kernel of a matrix $H$, then \(H = \Sigma B\) for some \(B \in \matrixring{\nu \times \mu}{\bfield}\) which has no zero row. 
\end{lemma}

\begin{proof}
The statement comes from the commutative diagram of $\bfield$-vector spaces
$$\xymatrix{0 \ar[r] & \mathcal{R}\rho \ar[r] & \mathcal{R} \ar[r]^-{\cdot H} & \mathcal{R}/\mathcal{R}\rho \ar[r] & 0 \\
0 \ar[r] & \mathcal{R}\lambda \ar[r] \ar@{^{(}->}[u] & \mathcal{R}\ar@{=}[u] \ar[r]^-{\cdot \Sigma} & \mathcal{R}/\mathcal{R}\lambda \ar[r] \ar@{-->}[u]_-{\cdot B}& 0 \\
}$$
If $\mathcal{R}\rho$ corresponds to the left kernel of a matrix $H$, there exists a surjective $\bfield$-linear map $\mathcal{R}/\mathcal{R}\lambda\to \mathcal{R}/\mathcal{R}\rho$ defined by right multiplication by a $(\nu\times \mu)$-matrix $B$, such that $\Sigma B=H$. Since $\mathcal{R}\rho$ is also the left kernel of $\Sigma E^\mu$, there exists a non singular $(\mu\times \mu)$-matrix $P$ such that $\Sigma E^\mu P=\Sigma B$. Since $\Sigma$ defines a surjective linear map, then $E^\mu P=B$. Finally, $B$ is obtained from $E^\mu$ by elementary operations on the columns. Since $E^\mu$ has no zero row, $B$ is so.
\end{proof}

There is a strong connection between \(\rho\) and \(\lambda\), since the error locator polynomial is minimal for \(\rho\) in the following sense. 

\begin{proposition}\label{lambdaminimalfullydecomposable}
Let $\lambda'\in\mathcal{R}$ be a fully $\beta$-decomposable polynomial which is a left multiple of $\rho$. Then $\lambda \mid_r \lambda'$.
\end{proposition}

\begin{proof}
By Proposition \ref{Vskewcode},  $\rho \mid_r \lambda$  and, by hypothesis, $\rho \mid_r \lambda'$. Then $\rho \mid_r (\lambda,\lambda')_r$. Moreover, by Lemma \ref{fullydecomposable}, $(\lambda,\lambda')_r$ is fully $\beta$-decomposable. Let us denote $\phi =(\lambda,\lambda')_r$. We claim that $\phi=\lambda$, which implies the statement. 

Indeed, $\mathcal{R}\lambda \subseteq \mathcal{R}\phi$, then $\mathcal{R}\phi$ corresponds to the left kernel of $\Sigma Q$, where $Q$ is a full rank matrix. Analogously, $\mathcal{R}\phi\subseteq \mathcal{R}\rho$, so there exists $Q'$ of full rank such that $\mathcal{R}\rho$ is the left kernel of $\Sigma QQ'$.
By Lemma \ref{Rrhocomputation}, $\Sigma QQ'=\Sigma B$, where $B$ has full rank and no zero row. Hence, $QQ'=B$, because $\Sigma$ defines a surjective linear map, and $Q$ has no zero row.

Since \(\phi \mid_r \lambda\), any $\beta$-root of \(\phi\) must be a $\beta$-root of $\lambda$, so that it belongs to the set \(\{\sigma^{k_1}(\beta), \dots, \sigma^{k_\nu}(\beta)\}\). Observe that, by \eqref{findingroots}, \(\sigma^{k_j}(\beta)\) is a \(\beta\)-root of \(\phi\) if and only if 
\[
\rank \left( \begin{array}{c|c} \Sigma Q & \sigma^{k_j}(\alpha)^{[\sigma]} \end{array}\right) = \rank \Sigma Q.
\]
Hence, by Lemma \ref{q=p=m}, \(\phi\) being fully $\beta$-decomposable implies \(\{\sigma^{k_1}(\beta), \dots, \sigma^{k_\nu}(\beta)\}\) is the set of \(\beta\)-roots of \(\phi\). Thus \(\phi = \lambda\). 
\end{proof}

At this point we have developed all the ingredients needed to compute the error locator polynomial, which completes the steps for designing the Peterson-Gorenstein-Zierler algorithm for skew cyclic codes, see Algorithm \ref{PGZ}. 

\begin{algorithm}[H]
\caption{\texttt{PGZ decoding algorithm}}\label{PGZ}
\begin{algorithmic}[1]
\REQUIRE A received transmission \(y = \left (y_0,\ldots ,y_{n-1}\right ) \in L^n\) with no more than \(t\) errors.
\ENSURE The error \(e = \left (e_0,\ldots ,e_{n-1}\right )\) such that \(y-e \in \mathcal{C}\)
\FOR{\(0 \leq i \leq 2t-1\)}
	\STATE $s_i\gets\sum_{j=0}^{n-1} y_j \norm{j}{\sigma^i(\beta)}$
\ENDFOR
\IF{\(s_i = 0\) for all \(0 \leq i \leq 2t-1\)}
	\RETURN \(0\).
\ENDIF
\STATE \(S^t \gets \left( \sigma^{-j}(s_{i+j})\sigma^{i}(\alpha) \right)_{0 \leq i \leq t, 0 \leq j \leq t-1}\)
\STATE Compute \label{rcefSt}
\[
\operatorname{rcef}(S^t) = \left(\begin{array}{c|c} \begin{array}{c} I_\mu \\ \hline a_0 \cdots a_{\mu-1} \\ \hline H' \end{array} & 0_{(t+1)\times(t-\mu)} \end{array}\right).
\] 
\STATE \(\rho = (\rho_0, \dots, \rho_\mu) \gets (-a_0, \dots, -a_{\mu-1}, 1)\) and $\rho_N \gets (\rho_0,\ldots ,\rho_{\mu},0,\ldots,0) N$.\label{alg:rhocomput} \label{startfind}
\STATE $\{k_1,\ldots ,k_\nu\}\gets \text{ zero coordinates of } \rho_N$ \label{firstzerofind}
\IF{$\mu \neq \nu$} \label{mu<>nu}
	\STATE $M_\rho \gets \begin{pmatrix} \rho_0 & \rho_1 & \ldots & \rho_\mu & 0 & \ldots & 0\\
	0 & \sigma(\rho_0) & \ldots & \sigma(\rho_{\mu-1}) & \sigma(\rho_{\mu}) & \ldots & 0\\
	 &  & \ddots &  &  & \ddots & \\
	0 &   \ldots  & 0 & \sigma^{n-\mu-1}(\rho_0)& \ldots & \ldots  &\sigma^{n-\mu-1}(\rho_\mu)
	\end{pmatrix}_{(n-\mu)\times n}$ \label{startfind}
	\STATE $N_\rho \gets M_\rho N$
	\STATE $H_\rho \gets \operatorname{rref}(N_\rho)$ \label{rrefNrho}
	\STATE \(H'\) $\gets$ matrix obtained removing all rows of $H_\rho$ different from $\varepsilon_i$ for any $i$.\label{removefind}
	\STATE $\{k_1,\ldots ,k_\nu\}\gets \text{ zero column coordinates of } H'$ \label{endfind}
\ENDIF	 
\STATE Solve the linear system \((x_1, \dots, x_\nu)\left (\Sigma^{\nu-1}\right )^\intercal = (\alpha s_0, \sigma(\alpha) s_1, \dots, \sigma^{\nu-1}(\alpha) s_{\nu-1})\)\label{computerrorvalues}
\RETURN \((e_0,\ldots ,e_{n-1})\) with $e_i=x_i$ for $i\in\{k_1,\ldots, k_\nu\}$, and zero otherwise.
\end{algorithmic}
\end{algorithm}

\begin{theorem}
Let \(\bfield\) be a field, \(\sigma \in \Aut{\bfield}\) of order \(n\) and \(\bfield^\sigma\) the invariant subfield. Let the set \(\{\alpha, \sigma(\alpha), \dots, \sigma^{n-1}(\alpha)\}\) be a normal basis of \(\bfield\) over \(\subfield\) and \(\beta = \alpha^{-1}\sigma(\alpha)\). Let \(\mathcal{R} = \bfield[x;\sigma]/\langle x^n-1\rangle\), \(g = \lclm{x-\beta, \dots, x-\sigma^{\delta-2}(\beta)}\) and  \(\mathcal{C}\) the skew RS code such that \(\tovector^{-1}(\mathcal{C}) = \mathcal{R}g\). Then Algorithm \ref{PGZ} correctly finds the error \(e = (e_0, \dots, e_{n-1})\) of any received vector if the number of non zero coordinates of \(e\) is \(\nu \leq t = \left\lfloor \frac{\delta-1}{2} \right\rfloor\).
\end{theorem}

\begin{proof}
After the initial settings, Line \ref{alg:rhocomput} computes a right divisor \(\rho = \sum_{i=0}^\mu \rho_i x^i\) of the error locator \(\lambda\) by Proposition \ref{Vskewcode} and Lemma \ref{rhocomputation}. 

By \eqref{findingroots}, Line \ref{firstzerofind} computes all the $\beta$-roots of \(\rho\). By Lemma \ref{degreelclm}, \(\nu = \mu\) if and only if \(\rho\) is fully $\beta$-decomposable. In this case, by Proposition \ref{lambdaminimalfullydecomposable}, \(\rho = \lambda\). 

If $\nu\not = \mu$, since \(\deg \rho = \mu\), the rows of \(M_\rho\) generate \(\mathcal{R}\rho\) as an \(\bfield\)--vector space, and the rows of \(N_\rho\) also generates \(\mathcal{R}\rho\) under the change of basis corresponding to \(N\). Since \(H_\rho\) is the reduced row echelon form of \(M_\rho\), then its rows are also a basis of \(\mathcal{R}\rho\) as an \(\bfield\)--vector space. By Lemma \ref{rreffullydecomposable}, the rows of \(H'\) generate an \(\bfield\)--vector subspace \(\mathcal{R}\lambda'\) for some fully $\beta$-decomposable polynomial \(\lambda'\). Since \(H'\) is obtained removing some rows of \(H_\rho\), it follows that \(\rho \mid_r \lambda'\). 

Let us now prove that $\lambda'=\lambda$. By Proposition \ref{lambdaminimalfullydecomposable}, \(\lambda \mid_r \lambda'\). Suppose that $\lambda \not= \lambda'$, then the matrix \(H_\lambda = \operatorname{rref}(M_\lambda N)\) contains an additional row $\varepsilon_d$ not in \(H'\). Since $\mathcal{R}\lambda\subseteq \mathcal{R}\rho$, $\rank \left (\begin{array}{c} H_\rho \\ \hline \varepsilon_d \end{array} \right )=\rank (H_\rho)$. By Lemma \ref{tech}, \(\varepsilon_d\) is a row of \(H_\rho\), so it is not removed in Line \ref{removefind} of Algorithm \ref{PGZ}. Hence \(\varepsilon_d\) belongs to \(H'\), a contradiction. Thus $\lambda=\lambda'$ and the error positions are computed. By Proposition \ref{errorvalues}, Line \ref{computerrorvalues} computes the error values.
\end{proof}

\begin{remark}
The complexity of Algorithm \ref{PGZ} is dominated by the computation of the reduced echelon forms in Lines \ref{rcefSt} and \ref{rrefNrho}. The theoretical efficiency of these lines belong to \(\mathcal{O}(t^3)\) and \(\mathcal{O}(n^3)\), respectively. Since, in the worst case, $t\approx n/2$, the complexity of Algorithm \ref{PGZ} is in $\mathcal{O}(n^3)$. Nevertheless, in most cases, the \textbf{if} part (Lines \ref{startfind}-\ref{endfind}) is not executed. Indeed, $\mu$ is not the true number of errors $\nu$ if and only if the determinant 
\[
|E^\nu|=\left | \begin{matrix} e_1 & \sigma^{-1}(e_1) & \cdots & \sigma^{-\nu+1}(e_1) \\
e_2 & \sigma^{-1}(e_2) & \ldots & \sigma^{-\nu+1}(e_2) \\
\vdots & \vdots & \ddots & \vdots \\
e_\nu & \sigma^{-1}(e_\nu) & \cdots & \sigma^{-\nu+1}(e_\nu)
\end{matrix} \right |=0
\]
Therefore, taking coordinates with respect to a fixed basis of $\bfield$ over $\bfield^\sigma$, we deduce that the set of errors $\{ e_1, \dots, e_\nu \}$ which yields $|E^\nu|=0$ is contained in the determinantal algebraic subvariety of $(\bfield^\sigma)^{\nu n}$ determined by the common zeroes of all $\nu \times \nu$ minors. The  dimension of this variety is known to be at most $n-\nu + 1$ (see, e.g., \cite[Exercise 10.10]{Eisenbud:1995}), which is strictly smaller than $\nu n=\mathrm{dim}_{\bfield^\sigma}\bfield^\nu$ if $\nu > 1$.   
\end{remark}

\section{Decoding skew RS codes}\label{secexam}

In this section we illustrate the scope of application of Algorithm \ref{PGZ} with some examples. The calculations have been made with the aid of the mathematical software Sagemath \cite{sage}.

\subsection{$\sigma$-Cyclic codes}

Our algorithm can be applied to some of the $\sigma$-cyclic codes defined in \cite{Boucher/Geiselmann/Ulmer:2007}. These are defined as left ideals of the factor algebra $\mathbb{F}_q[x;\sigma]/\langle x^n-1\rangle$, where $\mathbb{F}_q$ is the finite field of $q$ elements and the order of $\sigma$ divides $n$. So Algorithm \ref{PGZ} can be applied whenever $n$ is the order of $\sigma$, as the following example shows.

\begin{example}
Let $\mathbb{F}=\mathbb{F}_2(a)$ be the field with $2^{12}$ elements, where $a^{12}+a^7 + a^6 + a^5 + a^3 + a + 1=0$. Except for 0 and 1, we write the elements of $\mathbb{F}$ as powers of $a$. Consider $\sigma:\mathbb{F}\to \mathbb{F}$ defined by $\sigma=\tau^{10}$, where \(\tau\) is the Frobenius automorphism, i.e. $\sigma(a)=a^{1024}$. The order of $\sigma$ is 6, so a skew cyclic code over $\mathbb{F}$ is  a left ideal of the quotient algebra $\mathcal{R}=\mathbb{F}[x;\sigma]/\langle x^6-1\rangle$.
Let now take $\alpha=a$, which provides a normal basis of $\mathbb{F}$ as an $\mathbb{F}^\sigma$-vector space, and $\beta=\sigma(a)a^{-1}=a^{1023}$. Then, the images of $\beta$ under the powers of $\sigma$ give us the set $\{a^{1023},a^{3327},a^{3903},a^{4047},a^{4083},a^{4092}\}$. We may also construct the matrix of the change of basis
\[
N=\left(\begin{array}{cccccc}
1 & 1 & 1 & 1 & 1 & 1 \\
a^{1023} & a^{3327} & a^{3903} &
a^{4047} & a^{4083} & a^{4092} \\
a^{255} & a^{3135} & a^{3855}  &
a^{4035}  & a^{4080} & a^{1020}  \\
a^{63} & a^{3087} & a^{3843} & a^{4032} &
a^{1008}  & a^{252}  \\
a^{15} & a^{3075}  & a^{3840}  & a^{960}  & a^{240} & a^{60}  \\
a^{3} & a^{3072}  & a^{768}  & a^{192}   & a^{48} & a^{12} \end{array}\right).
\]
Let us consider the skew RS code generated by \[g=\lclm{x-a^{1023},x-a^{3327},x-a^{3903},x-a^{4047}}=x^4+a^{2103}x^3+a^{687}x^2+a^{1848}x+a^{759}.\] By Theorem \ref{MDScode}, it has Hamming distance 5, so, following  Algorithm \ref{PGZ}, it can correct until 2 errors. The reader may check easily that this code is not cyclic in the usual sense. Suppose then we need to send the message $m=x+a$, so the encoded polynomial to be transmitted is $c=mg=x^5+a^{3953}x^4+a^{1333}x^3+a^{2604}x^2+a^{1596}x+a^{760}$. After the transmission, we receive a polynomial $y=x^5+a^{3953}x^4+a^{671}x^3+a^{2604}x^2+a^{1596}x+a^{3699}$, i.e. there are two errors at positions 0 and 3. Actually, we have added the error $e=a^2+a^3x^3$.

We first calculate the full matrix of syndromes 
$$\left(\begin{array}{ccc}
a^{3170} & a^{2390}\\
a^{2645}  & a^{428}\\
a^{107} &  a^{248}
\end{array}\right),$$
and its reduced column echelon form
$$\left(\begin{array}{ccc}
1 & 0\\
0 & 1\\
a^{1950} &  a^{3315}
\end{array}\right).$$
Therefore, the rank of the matrix is two, and the monic polynomial in the left kernel is 
$\rho=x^2+a^{3315}x+a^{1950}$, and $\rho N=(0,a^{210},a^{2685},0,a^{1155},a^{3945})$. So $\rho$ is the error locator polynomial and the error positions are 0 and 3. Now, in order to compute the error values, we need to solve the system
$$(e_0,e_3)\left(\begin{array}{cc}
a & a^{1024}  \\
a^{64}  & a^{16} 
\end{array}\right)=\left(a^{3170} ,\,a^{2645} \right)
,$$
which yields $e_0=a^2$ and $e_3=a^3$, as expected.

Let us now suppose that we receive a polynomial $y=x^5+a^{3953}x^4+a^{671}x^3+a^{2604}x^2+a^{1596}x+a^{3699}$, that is, we have added the error $e=a^2+a^{1367}x^3$. In this case, the syndrome matrix and its reduced column echelon form are
$$\left(\begin{array}{ccc}
a^{59} & a^{65} \\
a^{1040} & a^{1046} \\
a^{2309} & a^{2315} 
\end{array}\right) \text{ and }
\left(\begin{array}{ccc}
1& 0 \\
a^{981} & 0 \\
a^{2250} & 0
\end{array}\right),
$$
respectively. Then $\rho=x+a^{981}$, and $\rho N=(a^{1437},a^{1281},a^{4053},a^9,a^{3149},a^{3853})$. Therefore, $\rho$ is not the error locator polynomial. We then compute the matrices
$$M_\rho=\left(\begin{array}{cccccc}
a^{981}  & 1 & 0 & 0
& 0 & 0 \\
0 & a^{1269} & 1
& 0 & 0 & 0 \\
0 & 0 & a^{1341}  & 1 & 0
& 0 \\
0 & 0 & 0 & a^{1359}
& 1 & 0 \\
0 & 0 & 0 & 0 & a^{3411} & 1
\end{array}\right)$$ and $$N_\rho=M_\rho N=\left(\begin{array}{rrrrrr}
a^{1437} & a^{1281}& a^{4053}& a^{9}& a^{3149}& a^{3853} \\
a^{2406}& a^{576}& a^{1845}& a^{978}& a^{1799}& a^{1984}\\
a^{3672}& a^{3471}& a^{1293}& a^{2244}& a^{3509}& a^{493}\\
a^{1941}& a^{3171}& a^{1155}& a^{513}& a^{1889}& a^{1144}\\
a^{2532}& a^{3096}& a^{3168}& a^{1104}& a^{1484} & a^{283}\\
\end{array}\right).$$
Now, the reduced row echelon form of $N_\rho$ is as follows      	
$$H_\rho=\left(\begin{array}{cccccc}
1 & 0 & 0 & a^{2667}  & 0 &
0 \\
0 & 1 & 0 & 0 & 0 & 0 \\
0 & 0 & 1 & 0 & 0 & 0 \\
0 & 0 & 0 & 0 & 1 & 0 \\
0 & 0 & 0 & 0 & 0 & 1
\end{array}\right).$$ If we remove the first row, we find zero columns at positions 0 and 3, i.e. the error positions.  Finally, we solve the system 
$$(e_0,e_3) \left(\begin{array}{cc}
a & a^{1024}  \\
a^{64}  & a^{16} 
\end{array}\right)=\left(a^{59} ,\,a^{1040} \right)
,$$
which yields $e_0=a^2$ and $e_3=a^{1367}$.
\end{example}

\subsection{Skew cyclic convolutional codes}

In \cite{GLNSCCC} it is taken into consideration a novel approach to cyclicity for convolutional codes by introducing the so-called skew cyclic convolutional codes (SCCCs). This perspective considers the embedding of a polynomial ring $\mathbb{F}[z]$ into its field of fractions $\mathbb{F}(z)$, so that SCCCs are skew cyclic codes over \(\field[](z)\).

\begin{example}
Let $\mathbb{F}=\mathbb{F}_4=\{0,1,a,a^2\}$ be the field with four elements, $\mathbb{F}(z)$ the field of rational functions over $\mathbb{F}$ and $\sigma:\mathbb{F}(z)\to \mathbb{F}(z)$ the automorphism of order 5 defined by $\sigma(z)=(z+a)/(z+a^2)$. The working ring is  $\mathcal{R}=\mathbb{F}(z)[x;\sigma]/\langle x^5-1 \rangle$. Let us consider $\alpha=z\in \mathbb{F}(z)$, $\beta=\alpha^{-1}\sigma(\alpha)=(z + a)/(z^2 + a^2z)$, and  the skew cyclic code generated by $g=\lclm{x-\beta,x-\sigma(\beta),x-\sigma^2(\beta),x-\sigma^3(\beta)}$. Concretely, 
{\small{$$g=x^{4} + \left(\tfrac{z + a}{z^{5} + a^2 z}\right) x^{3} +
\left(\tfrac{a z^{5} + a^2 z^{4} + a z + a^2}{z^{5} +
a^2 z^{4} +a^2 z + a}\right) x^{2} +
\left(\tfrac{a^2 z^{5} + z^{4} + z + a}{z^{4} + a^2
}\right) x + \tfrac{a z^{5} + a^2 z^{4}}{a^2 z^{5} + a^2 z^{4} + a z + a}.$$}}
So it can correct up to two errors. Suppose that $g$ is transmitted and we receive the polynomial $y$ given by
{\small{$$x^{4} + \left(\tfrac{1}{z^{4} + a^2}\right) x^{3} +
\left(\tfrac{a z^{5} + a^2 z^{4} + a z + a^2}{z^{5} +
a^2 z^{4} +a^2 z + a}\right) x^{2} +
\left(\tfrac{a^2 z^{6} + z^{5} + z^2 + az+1}{z^{5} + a^2z
}\right) x + \tfrac{a z^{5} + a^2 z^{4}}{a^2 z^{5} + a^2 z^{4} + a z + a},$$}}
i.e., y=g+e, where $e=a/(z^5+a^2z)x^3+1/(z^5+a^2z)x$. We follow Algorithm \ref{PGZ} and compute the matrix of syndromes,
$$\left(\begin{array}{cc}
\frac{a^2 z^{2} + a z + a^2}{a^2 z^{7} +
a^2 z^{6} + a^2 z^{5} + a z^{3} + a z^{2}
+ a z} & \frac{a z^{7} +a^2 z^{6} + a^2 z^{5} + a z^{4} + a z^{3} + a^2 z^{2} + a^2 z + a}{z^{3} + a z^{2} + a z + a^2} \\
\frac{z^{2} + z + a^2}{a z^{7} + a z^{6} + z^{3} + z^{2}} &
\frac{a z^{6} + a z^{5} + z^{4} + a z^{2} + a z + 1}{a z^{2} + z} \\
\frac{a^2 z^{2} + z + a^2}{a z^{6} + a^2
z^{5} + z^{2} + a z} & \frac{a z^{7} + z^{6} + z^{5} + a z^{4} + a
z^{3} + z^{2} + z + a}{a^2 z^{2} + a^2 z +
a^2}\end{array}\right)$$
and its reduced column echelon form
$$\left(\begin{array}{cc}
1 & 0 \\
\frac{a^2 z^{4} + a z^{2} + z + a}{z^{4} + a z^{3} + a
z^{2} + z} & 0 \\
\frac{a z^{3} + a z^{2} + 1}{z^{2} + a^2 z + 1} & 0
\end{array}\right).$$
Therefore,
$$\rho=       	
x + \tfrac{a^2 z^{4} + a z^{2} + z + a}{z^{4} + a z^{3} + a
z^{2} + z} \equiv \left (\tfrac{a^2 z^{4} + a z^{2} + z + a}{z^{4} + a z^{3} + a
z^{2} + z},1,0,0,0\right).$$
In this case, the matrix $N$ is given as follows:
$$N=\left(\begin{array}{ccccc}
1 & 1 & 1 & 1 & 1 \\
\frac{z + a}{z^{2} + a^2 z} & \frac{a^2 z^{2} + z + a}{a z^{2} + a^2 z} &
\frac{z}{z^{2} + a^2 z + a} & \frac{z^{2} + z +
1}{a^2 z + a^2} & \frac{z^{2} + z}{a^2 z + a} \\
\frac{a z + a}{ z^{2}} & \frac{a^2 z + a}{a z^{2} + 1} & \frac{z^{2} + a^2
z}{a^2 z^{2} + a^2} & a^2 z^{2} + z
& \frac{z^{2} + a^2 z + a}{a^2 z^{2} +
1} \\
\frac{a^2}{a z^{2} + a^2 z} & \frac{a^2 z^{2} + 1}{z^{2} + a^2 z + a} &
\frac{z^{2}}{a z + a} & \frac{a^2 z^{2} + a}{z + a^2} & \frac{a^2 z^{2} + a^2}{z^{2} + a^2 z} \\
\frac{a^2 z + a}{z^{2} + z} & \frac{a^2 z^{2} + a z}{a^2 z + 1} & \frac{z^{2} + a
z}{a z^{2} +a^2 z + 1} & \frac{a z^{2} + z + a^2}{a z} & \frac{a^2 z + a^2}{z^{2} + z + 1}
\end{array}\right).$$
Now, $\rho  N$ has no zero coordinate, so $\rho$ is not the error locator polynomial. Following Algorithm \ref{PGZ},
$$M_\rho=\left(\begin{array}{ccccc}
\frac{a^2 z^{4} + a z^{2} + z + a}{z^{4} + a z^{3} + a
z^{2} + z} & 1 & 0 & 0 & 0 \\
0 & \frac{a z^{4} + z^{3} + z^{2} + a z}{a^2 z^{3} +
a z^{2} + a^2 z + a^2} & 1 & 0 & 0 \\
0 & 0 & \frac{a^2 z^{3} + a z^{2} + a}{z^{4} +
z^{2} + a^2 z + a^2} & 1 & 0 \\
0 & 0 & 0 & \frac{a^2 z^{4} + a^2 z^{3} +a^2 z^{2} + a z + 1}{a^2
z^{3} + a^2 z^{2} + z} & 1
\end{array}\right).$$
Hence, the reduced row echelon form of $M_\rho N$ is as follows:
$$H_\rho=\left(\begin{array}{ccccc}
1 & 0 & 0 & 0 & 0 \\
0 & 1 & 0 & \frac{a z^{2} + 1}{z + a^2} & 0 \\
0 & 0 & 1 & 0 & 0 \\
0 & 0 & 0 & 0 & 1
\end{array}\right).$$
If we remove the second row, the resultant matrix has a zero column at positions 1 and 3. Finally, we may find the error values by solving the linear system
$$(e_1,e_3) \left(\begin{array}{rr}
\frac{z + a}{z + a^2} & \frac{a z + a}{z}
\\
\frac{a }{ z + a} & \frac{a^2z + a}{z + 1}
\end{array}\right)=\left(\tfrac{a^2 z^{2} + a z + a^2}{a^2
z^{7} + a^2 z^{6} + a^2 z^{5} + a z^{3} +
a z^{2} + a z},\,\tfrac{z^{2} + z + a^2}{a z^{7} + a z^{6} + z^{3} +
z^{2}}\right),$$
which yields $e_1=1/(z^5 +a^2 z)$ and $e_3=a/(z^5 + a^2z)$.
\end{example}

\subsection{Skew cyclic codes over a cyclotomic field}

Here we show an additional example of a class of skew cyclic codes over a non-conventional field. The base field of this kind of codes is a cyclotomic field $\mathbb{Q}(\chi)$, where $\chi$ is an $n$th root of unity.

\begin{example}\label{cyclotomic}
Let $L = \mathbb{Q}(\chi)$, where $\chi$ is a primitive 7th root of unit, and $\sigma: L \rightarrow L$ defined by $\sigma(\chi) = \chi^3$.  In this case, the order of $\sigma$ is $6$. Let us set $\alpha=\chi$, $\beta=\alpha^{-1}\sigma(\alpha)=\chi^2$ and $\delta=5$, so that the corresponding skew RS code is generated by 
$$g =2 x^{4} + \left(-\chi^{5} - \chi^{3} - \chi^{2}\right) x^{3} +
\left(\chi^{3} + \chi + 1\right) x^{2} + \left(\chi^{5} + \chi^{4} +
1\right) x + \chi^{5} - \chi^{2} + \chi + 1.
$$  
Suppose that $g$ is transmitted, i.e. the message $m=1$ is sent, and we receive the polynomial
$$y=2 x^{4} + \left(-\chi^{5} - \chi^{3} - \chi^{2}\right) x^{3} +
\left(\chi^{3} + 2 \chi + 1\right) x^{2} + \left(\chi^{5} + \chi^{4} +
1\right) x + \chi^{5} - \chi^{2} + \chi + 1.$$
Since $t=2$, the syndrome matrix $S^t$, and its reduced column echelon form, are given by 
$$\left(\begin{array}{cc}
\chi^{3} & 1 \\
1 & \chi^{4} \\
\chi^{5} & \chi^{2}
\end{array}\right) \text{ and } 
\left(\begin{array}{cc}
1 & 0 \\
\chi^{4} & 0 \\
\chi^{2} & 0
\end{array}\right),
$$
respectively. Therefore $\mu=1$ and $\rho=x-\chi^4$. Now, the matrix of norms is
$$\left(\begin{array}{cccccc}
1 & 1 & 1 & 1 & 1 & 1 \\
\chi^{2} & b
& \chi^{4} & \chi^{5} & \chi & \chi^{3} \\
\chi & \chi^{3} & \chi^{2} & b & \chi^{4} & \chi^{5} \\
\chi^{5} & \chi & \chi^{3} & \chi^{2} & b & \chi^{4} \\
\chi^{3} & \chi^{2} &b & \chi^{4} & \chi^{5} & \chi \\
\chi^{4} & \chi^{5} & \chi & \chi^{3} & \chi^{2} &
b
\end{array}\right),$$
where $b=-\chi^{5} - \chi^{4} - \chi^{3} - \chi^{2} - \chi - 1$, so
$$\rho_N=\left(-\chi^{4} + \chi^{2},\,-\chi^{5} - 2 \chi^{4} - \chi^{3} -
\chi^{2} - \chi - 1,\,0,\,\chi^{5} - \chi^{4},\,-\chi^{4} +
\chi,\,-\chi^{4} + \chi^{3}\right).$$
By Algorithm \ref{PGZ}, there is a single error at position 2 whose value may be computed by solving the equation $e_1\chi^2=\chi^3$. That is, the error is $e=\chi x^2$.
\end{example}

\section*{Funding}

Research supported by grants MTM2016-78364-P, MTM2013-41992-P and TIN2013-41990-R
from Ministerio de Econom\'ia y Competitividad and from Fondo Europeo de
Desarrollo Regional FEDER.



\end{document}